\documentclass[11pt,a4paper]{article}
\usepackage[utf8]{inputenc}

\usepackage{amssymb,amsfonts,amsmath,amsthm,bm,mathrsfs,jheppub}
\usepackage{titletoc,paralist,color,graphicx,tikz}

\usepackage{hyperref}
\usepackage{epstopdf}
\usepackage{graphicx}

\usetikzlibrary{calc,shapes.geometric}

\definecolor{darkred}{RGB}{173,34,48}
\usetikzlibrary{math,shapes.geometric}

\def\rr{\mathbb{R}}
\def\cc{\mathbb{C}}
\def\zz{\mathbb{Z}}
\def\dd{\mathrm{d}}
\def\ii{\mathrm{i}}

\theoremstyle{definition}
\newtheorem{defi}{Definition}[part]

\theoremstyle{plain}

\newtheorem{pro}[defi]{Proposition}

\date{\today}
\title{Moduli Space of Paired Punctures, Cyclohedra and Particle Pairs on a Circle} 

\begin{document}
\author[a,b]{Zhenjie Li}
\author[a,b]{,\,Chi Zhang}
\affiliation[a]{CAS Key Laboratory of Theoretical Physics, Institute of Theoretical Physics, Chinese Academy of Sciences, Beijing 100190, China}
\affiliation[b]{School of Physical Sciences, University of Chinese Academy of Sciences, No.19A Yuquan Road, Beijing 100049, China}
\emailAdd{lizhenjie@itp.ac.cn}
\emailAdd{zhangchi@itp.ac.cn}
\date{\today}
\abstract{In this paper, we study a new moduli space $\mathcal{M}_{n+1}^{\mathrm{c}}$, which is obtained from $\mathcal{M}_{0,2n+2}$ by identifying pairs of punctures. We find that this space is tiled by $2^{n-1}n!$ \emph{cyclohedra}, and construct the canonical form for each chamber. We also find the corresponding Koba-Nielsen factor can be viewed as the potential of the system of $n{+}1$ pairs of particles on a circle, which is similar to the original case of $\mathcal{M}_{0,n}$ where the system is $n{-}3$ particles on a line. We investigate the intersection numbers of chambers equipped with Koba-Nielsen factors. Then we construct cyclohedra in kinematic space and show that the scattering equations serve as a map between the interior of worldsheet cyclohedron and kinematic cyclohedron. Finally, we briefly discuss string-like integrals over such moduli space.}

	\maketitle
	\section{Introduction}

The scattering processes are closely linked to the Riemann surfaces with punctures since the birth of string theory. In particular, the scattering of $n$ massless particles can be described as a \emph{localized} integral over some moduli spaces in the context of Witten-RSV formalism~\cite{Witten:2003nn,Roiban:2004yf} or CHY formalism~\cite{Cachazo:2013gna,Cachazo:2013hca,Cachazo:2013iea}. In other words, the tree-level $S$-matrix of massless particles can be computed based on a map from some moduli space to kinematic space, which in general dimension is provided by the so-called scattering equations
\[
\sum_{j\neq i}\frac{k_{i}\cdot k_{j}}{z_{i}-z_{j}} = 0    \qquad\text{for $i\in\{1,\ldots,n\}$}. 
\]   
This map itself has been studied in detail in a new framework~\cite{Arkani-Hamed:2017mur} and recast as a pushforward from differential forms on the moduli space $\mathcal{M}_{0,n}$, which are called as \emph{worldsheet forms}, to differential forms on the kinematic space of $n$ massless particles, which are called as \emph{scattering forms}. The combinatorial and geometrical aspects of moduli space and kinematic space become hence crucial in this context. Such idea has been further developed by considering extended logarithmic differential forms as in~\cite{He:2018pue} and introducing subspace in kinematic space as in~\cite{He:2018pue,Salvatori:2018aha}.

In another closely-related approach to amplitudes---intersection theory~\cite{Mizera:2017cqs,Mizera:2017rqa}, the combinatorial and geometrical properties of moduli space also play an important role. In this formalism, scattering amplitudes and related physical qualities are understood as intersection numbers of various twisted cycles and(or) cocycles which characterize the boundary information of the moduli space concerned. In the case of $\mathcal{M}_{0,n}$, a basis of twisted cocycles turn out to be the Parke-Taylor forms which are precisely the most important worldsheet forms considered in~\cite{Arkani-Hamed:2017mur}.  The combinatorial structure of some physical qualities, for example, the planar bi-adjoint $\phi^{3}$ amplitudes in~\cite{Cachazo:2013iea,Arkani-Hamed:2017mur} and the inverse of KLT matrix $m_{\alpha^{\prime}}(\alpha\vert\beta)$ in~\cite{Mizera:2016jhj,Mizera:2017cqs}, are governed by the boundary structure of $\mathcal{M}_{0,n}(\mathbb{R})$, which turn out to be the so-called \emph{assocaihedra}~\cite{stasheff1963homotopyI,stasheff1963homotopyII}.

In this article, we will use these two tools---positive geometry~\cite{Arkani-Hamed:2017tmz} and intersection theory---to study another interesting moduli space, which is obtained from $\mathcal{M}_{0,2n+2}$ by imposing identifications for pairs of particles, such as $z_{i}=-z_{i+n+1}$, equipped with the corresponding Koba-Nielsen factor. This is a $n$-dimensional space and denoted as $\mathcal{M}_{n+1}^{\mathrm{c}}$ in the following, and the corresponding Koba-Nielsen factor can be inherited from the usual Koba-Nielsen factor for $2n+2$ particles by imposing the above identifications. In particular, the real moduli space $\mathcal{M}_{n+1}^{\mathrm{c}}(\mathbb{R})$ is tiled by $n$-dimensional \emph{cyclohedra} $W_{n}$. The main result of this article is to recover this structure in kinematic space via intersection theory and positive geometry. To this end, we first introduce the ingredients in section \ref{sec2}, i.e., Parke-Taylor forms and scattering equations, and study their properties further. To do that, We will give the precise definition of $\mathcal{M}_{n+1}^{\mathrm{c}}(\mathbb{R})$ and $\mathcal{M}_{n+1}^{\mathrm{c}+}(\mathbb{R})$ and show:  (i) $\mathcal{M}_{n+1}^{\mathrm{c}}(\mathbb{R})$ is tiled by $2^{n-1}n!$ $\mathcal{M}_{n+1}^{\mathrm{c}+}(\mathbb{R})$ which can be identified with a cyclohedron $ W_{n}$, (ii) the corresponding Parke-Taylor form is the canonical form on $\mathcal{M}_{n+1}^{\mathrm{c}+}(\mathbb{R})$,
and (iii) the scattering equations as the stationary points of a particle system have $2^{n-1}n!$ solutions. 

With this knowledge of moduli space $\mathcal{M}_{n+1}^{\mathrm{c}}$, we move to the kinematic space. The way through intersection theory is quite straightforward, the intersection number of two twisted cocycles $\alpha$ and $\beta$ tell us the result and the structure of $m_{\alpha^{\prime}}(\alpha\vert\beta)$, and shown in section \ref{sec3}. While in section \ref{sec4}, we will first construct a cyclohedron in the kinematic space, then rewriting the scattering equations as a map from the moduli cyclohedron to the kinematic cyclohedron, as done in~\cite{Arkani-Hamed:2017mur}. As a result of this pushforward, the amplitudes $m(\alpha\vert \beta)$ can be extracted from the scattering form by pullbacking to the corresponding subspace. It becomes obvious that $m(\alpha\vert\beta)$ is given by $m_{\alpha^{\prime}}(\alpha\vert\beta)$ in the limit of $\alpha^{\prime}\to 0$.

Last but not the least, we will consider the natural integrals on such moduli space that are analogs of the so-called $Z$-integrals for open-string case~\cite{Mafra:2016mcc}.  Interestingly, the Koba-Nielsen factor for integrals of open strings and one closed-string insertion is a special limit of our case. However, this limit needs to be taken carefully since some integrals don't behave well under this limit. Here we just give a brief survey on its $\alpha^{\prime}$-expansion and number theoretical properties before taking this limit.

	\section{Moduli Space of Pairs of Punctures on the Riemann Sphere} \label{sec2}

In this section, we consider the moduli space of distinct pairs of punctures $(z,z')$ on the Riemann sphere $\mathbb{CP}^1$, 
where two punctures in each pair are related by $z'=-z$. 

It's clear that an arbitrary $\operatorname{SL}(2,\cc)$ transformation on the Riemann sphere 
will not keep two punctures $z$ and $z'$ in a pair still satisfying the relation that $z'=-z$. One can check that the surviving transformation is
\begin{equation*}
    z\mapsto \lambda z \quad\text{and}\quad 
    z\mapsto \frac{1}{z},
\end{equation*}
where $\lambda$ is a non-zero complex number. Therefore, the remaining $\operatorname{SL}(2,\cc)$ 
effectively is the group $(\cc-\{0\}) \rtimes \zz_2$ whose elements are transformations
\begin{equation}
    (\lambda,a):z\mapsto \lambda z^{a},
\end{equation}
where $\lambda\in (\cc-\{0\})$ and $a=\pm 1$.\footnote{ 
It's the simi-product of $(\cc-\{0\})$ and $\zz_2$ since the multiplication rule is $(\lambda_1,a_1)\cdot (\lambda_2,a_2)=(\lambda_1\lambda_2^{a_1},a_1a_2)$.}

Now, suppose there are $n+1$ distinct such pairs of two punctures
$\{(z_i,z'_i)\,:\, 0\leq i\leq n\}$ on $\mathbb{CP}^1$. 
Our complex moduli space is the quotient of this configurations space, modulo the remaining redundancy $(\cc-\{0\}) \rtimes \zz_2$. 
Let's denote it by $\mathcal M^{\mathrm{c}}_{n+1}(\mathbb C)$. 
For future convenience, it's useful to enlarge the
index set to $\zz_{2n+2}=\{0,\dots,n,n+1,\dots,2n+1\}$ and
identify $z'_i$ with $z_{i+n+1}$ or $z_{\tilde \imath}$,
where $\tilde \imath=i+n+1\in \zz_{2n+2}$.

\subsection{Real Moduli Space \texorpdfstring{$\mathcal{M}^{\mathrm{c}}_{n+1}(\mathbb{R})$}{M(n+1)} and Cyclohedron \texorpdfstring{$W_{n}$}{W(n)}} \label{sec2.1}

In the context of string theory, the moduli space $\mathcal M_{0,n}(\cc)$ of
$n$ distinct punctures on the Riemann sphere $\mathbb{CP}^1$ and its real part $\mathcal M_{0,n}(\rr)$ of $n$ distinct punctures 
on the circle (i.e. the boundary of the disk) $\mathbb{RP}^1=S^1$
are used to calculate the tree-level amplitude of closed strings and open strings respectively. 
Similarly, we can consider the real part of our complex moduli space $\mathcal M_{n+1}^{\mathrm{c}}(\mathbb C)$
by setting all these pairs of punctures on the unit circle of $\cc$. 
We denote this new real moduli space by $\mathcal{M}^{\mathrm{c}}_{n+1}(\mathbb R)$.
As in the complex case, the $\operatorname{SL}(2,\mathbb{R})$ redundancy now
reduces to $\operatorname{U}(1)\rtimes \zz_2$, which is generated by
\[
    z\mapsto \exp(\ii\theta)z
    \quad\text{and}\quad
    z\mapsto 1/z=\overline{z},
\]
where $\theta\in \rr$ and $|z|=1$. Since all punctures now live on the unit circle, it's very convenient to use their arguments 
$\theta_i$ to coordinate them, i.e., $z_i=\exp(\ii \theta_i)$. 
Besides, it's also convenient to introduce another set of useful real variables
\[
    x_i=\tan(\theta_i/2)=\ii\, \frac{1-z_i}{1+z_i}.
\]
In fact, it's the natrual coordinates of another possible great circle $\mathbb{R}\cup \{\infty\}$ on the Riemann sphere where pairs of punctures can live.

A lot of works on $\mathcal M_{0,n}(\rr)$ has revealed some 
hidden structures of this moduli space. For example, 
its Fulton-MacPherson compactification can be tiled by \emph{associahedra}~\cite{devadoss1999tessellations}.
In the rest of this subsection, we will show some similar structures on $\mathcal{M}^{\mathrm{c}}_{n+1}(\mathbb R)$.

Firstly, we fix the $\operatorname{U}(1)\rtimes \zz_2$ redundancy
by setting $z_0=-z_{\tilde 0}=1$ and limiting $z_1$ on the upper semi-circle $S^1_+$. 
Since either $z$ or $-z$ is on the upper semi-circle for any pairs of punctures $(z,-z)\neq (1,-1)$ on the unit circle,
one can use the one living on the $S^1_+$ to represent this pair.
Therefore, after the gauge fixing, 
$\mathcal M^{\mathrm{c}}_{n+1}(\mathbb R)$ is described by
\[
    \mathcal M^{\mathrm{c}}_{n+1}(\mathbb R)=\{(z_1,\pm z_2,\dots, \pm z_n)
    \in (S_+^1)^n\}-\Delta,
\]
where 
\[
    \Delta=\{(z_1,\pm z_2,\dots, \pm z_n)\in (S_+^1)^n\,:\, 
    \exists\, i,j \text{ s.t. } z_i^2=z_j^2 \text{ or }
z_i^2=1\}.
\]
One should exclude $\Delta$ because that points in this subset represent some pairs of punctures on the unit circle collide.

Since there are two choices $z_i$ and $-z_i$ for all $2\leq i\leq n$, 
$\mathcal M_{n+1}^{\mathrm{c}}(\mathbb R)$ is a $2^{n-1}$ copies of the following space
\[
    \{(p_1,p_2,\dots, p_n)\in (S_+^1)^n\}-\Delta
    =\{(1,p^2_1,p^2_2,\dots, p^2_n)\in (S^1)^{n+1}\}-\Delta',
\]
where 
\[
    \Delta'=\{(1,q_1,q_2,\dots, q_n)\in (S^1)^{n+1}\,:\,
     \exists\, i,j \text{ s.t. $q_i=q_j$ or $q_i=1$}\},
\]
This is a known moduli space $\mathrm{F}(S^1,n+1)/S^1$ which has been well investigated in \cite{devadoss2002space}. Therefore,
our real moduli space $\mathcal M_{n+1}^{\mathrm{c}}(\mathbb R)$ is 
$
    2^{n-1} \mathrm{F}(S^1,n+1)/S^1.
$

The real Fulton-MacPherson compactification of
$\mathrm{F}(S^1,n+1)/S^1$, denoted by $\overline{\mathcal Z}^{n+1}$, can be tiled by $n!$
\textit{cyclohedra} \cite{devadoss2002space}, so the same compactification of our moduli space
can be tiled by $2^{n-1}n!$ cyclohedra.
Suppose $\mathsf{proj}:\overline{\mathcal Z}^{n+1}\to \mathrm{F}(S^1,n+1)/S^1$ 
is the compactification map, every cyclohedron is given by
\[
    W_{n}(\sigma)=
    \overline{\mathsf{proj}^{-1}\left(
    \{0<\varphi_{\sigma(1)}<\varphi_{\sigma(2)}<\cdots<\varphi_{\sigma(n)}<2\pi\}    
    \right)},
\]
where $\varphi_i$ is the central angle of $p^2_i$ and $\sigma\in S_n$ is a permutation of indices set $\{1,\ldots,n\}$. In our case, every cyclohedron is given by
\begin{align}
W_{n}(\sigma,\{\mathsf{s}_{i}\})&=
\overline{\mathsf{proj}^{-1}\left(
\{0<\theta^{\mathsf{s}_{\sigma(1)}}_{\sigma(1)}<\theta^{\mathsf{s}_{\sigma(2)}}_{\sigma(2)}<\cdots<\theta^{\mathsf{s}_{\sigma(n)}}_{\sigma(n)}<\pi\}
\right)}, \noindent \\
&=:\overline{\mathsf{proj}^{-1}\left(\mathcal{M}_{n+1}^{\mathrm{c}+}(\sigma,\{\mathsf{s}_{i}\})\right)}
\end{align}
where $\mathsf{s}_i=\pm 1$ and 
\[
    \theta^+_i=\theta_i,\quad \theta^-_i=\theta_i-\pi.
\]

From this construction, we can conclude that facets of the cyclohedron $\mathcal{M}_{n+1}^{\mathrm{c}+}(\sigma,\{\mathsf{s}_{i}\})$
(before compactification) lay on hyperplanes defined by equations
$\mathsf{s}_{\sigma(i)} z_{\sigma(i)}-\mathsf{s}_{\sigma(i+1)} z_{\sigma(i+1)}=0$.
In terms of $\{\theta_i\}$, these hyperplane equations are
\[
    0\leq \theta_i < 2\pi \quad \text{and}\quad
    \theta_i-\theta_j=\begin{cases}
        0,       & \text{if $\mathsf{s}_i-\mathsf{s}_j=0 \,\bmod \,2$,}
        \\
        \pm \pi, & \text{if $\mathsf{s}_i-\mathsf{s}_j=\pm 1$.}
    \end{cases}
\]

For example, hyperplanes for $n=2$ are shown in the following diagram:
\begin{center}
    \begin{tikzpicture}[scale=2]
        \filldraw[fill=green!20] (0,0) -- (0,2) -- (1,2) -- (1,0)
        -- (0,0);
        \draw[thick,->] (0,0) -- (2.5,0) node[right] {$\theta_1$};
        \draw[thick,->] (0,0) -- (0,2.5) node[above] {$\theta_2$};
        \draw[thick] (0,0) -- (2,2);
        \draw[thick] (0,0) node[below] {$0$}
        -- (1,0) node[below=2pt] {$\pi$}
        -- (2,0) node[below] {$2\pi$}
        -- (2,2)
        -- (0,2) node[left] {$2\pi$}
        -- (0,1) node[left] {$\pi$}
        -- (0,0);
        \draw[thick] (1,0) -- (2,1) -- (0,1)
        -- (1,2) -- (1,0);
    \end{tikzpicture}
\end{center}
with the identification $(\theta_1,\theta_2)\sim (2\pi-\theta_1,2\pi-\theta_2)$ due to the $\zz_2$ redundancy.

As before, we can fix this redundancy by setting $\theta_1 < \pi$,
so there are $4=2^1\times 2!$ chambers (green ones in the above
diagram) corresponding to the cyclohedron before compactification.
After compactification, each chamber becomes a hexagon, the 2-dimensional cyclohedron.
For example, the triangle $0< \theta_1 <\theta_2<\pi$ actually becomes
\begin{center}
    \begin{tikzpicture}[scale=2]
        \draw[thick] (0,0.2) -- (0,1.8) -- (0.2,2) --
        (1.8,2) -- (2,1.8) -- (0.2,0) -- (0,0.2);
        \draw (1,1) node[below=10pt,right=10pt] {$\theta_1^+=\theta_2^+$};
        \draw (1,2) node[above] {$\theta_2^-=0$};
        \draw (0,1) node[left] {$\theta_1^+=0$};
        \draw[thick,red] (0,1.8) -- (0.2,2);
        \draw[thick,red] (1.8,2) -- (2,1.8);
        \draw[thick,red] (0.2,0) -- (0,0.2);
        \draw[thick,red,->] (-0.3,2.1) node[left,black] {$\theta_1^+=\theta_2^-=0$}
        .. controls (-0.2,2.2) and (-0.1,2.1) .. (0.05,1.95);
        \draw[thick,red,->] (-0.3,-0.1) node[left,black] {$\theta_1^+=\theta_2^+=0$}
        .. controls (-0.2,-0.2) and (-0.1,-0.1) .. (0.05,0.05);
        \draw[thick,red,->] (2.3,2.1) node[right,black] {$\theta_1^-=\theta_2^-=0$}
        .. controls (2.1,2.2) and (2.05,2.1) .. (1.95,1.95);
    \end{tikzpicture}
\end{center}
where three red facets are the result of compactification.
Notice that, every vertex should be blowed-up in the compactification such that no two cyclohedra intersect at a point. 

For general $n$, we can construct $W_{n}(\sigma,\{\mathsf{s}_{i}\})$ from the corresponding $\mathcal{M}_{n+1}^{\mathrm{c}+}(\sigma,\{\mathsf{s}_{i}\})$ and the compactification of the whole space $\mathcal{M}^{\mathrm{c}}_{n+1}(\mathbb R)$ by the following steps:
\begin{enumerate}[\quad (1)]
    \item For every chamber before compactification, 
        label $(ij)$ on 
        each codimension-one face (facet) $H$ defined by $z_{i}=z_j$. 
    \item For the face $H$ with higher codimension $k$ defined 
        by $z_{i_1}=\cdots=z_{i_k}$ of these chambers, 
        truncate it and label 
        $(i_1i_2\cdots i_k)$ on the new created facet. 
        Repeat this operation from low-dimensional faces to high-dimensional faces. 
    \item Glue the faces with the same label in different cyclohedron,
        they will be the same face after compactification.
\end{enumerate}

Here we give a direct combinatorial definition of cyclohedron. 
For any polytope $P$, there's a natural partial order of faces on it: for 
face $a$ and $b$,
\[
    a\leq b \quad \text{if and only if}\quad a\subset b,
\]
then $P$ defined a partially ordered set (or poset for short) $(P,\leq)$. 
Conversely, we can use a poset to define a polytope whose associated
poset of faces is isomorphic to the given poset. 

\begin{defi}[Cyclohedron]\label{def-1}
Let $\operatorname{Cyc}(n)$ be the poset of all centrally symmetric dissections 
of a convex $(2n+2)$-gon using non-intersecting diagonals, 
ordered such that $b \leq  b'$ if  $b$ is obtained from $b'$ by 
adding (non-intersecting) diagonals. The $n$-dimensional cyclohedron $W_n$ is 
convex polygon whose face poset is isomorphic to $\operatorname{Cyc}(n)$. 
\end{defi}

For example, $W_2$ is a hexagon.
\begin{center}
    \def\fp#1#2{
        \begin{tikzpicture}
            \node[draw,minimum size=0.6cm,regular polygon,regular polygon sides=6] (a) {};
            \draw (a.corner #1) -- (a.corner #2);
        \end{tikzpicture}
    }
    \def\fpo#1#2#3#4{
        \begin{tikzpicture}
            \node[draw,minimum size=0.6cm,regular polygon,regular polygon sides=6] (a) {};
            \draw (a.corner #1) -- (a.corner #2);
            \draw (a.corner #3) -- (a.corner #4);
        \end{tikzpicture}
    }
    \def\fpoo#1#2#3#4{
        \begin{tikzpicture}
            \node[draw,minimum size=0.6cm,regular polygon,regular polygon sides=6] (a) {};
            \draw (a.corner #1) -- (a.corner #2);
            \draw (a.corner #3) -- (a.corner #4);
            \draw (a.corner #1) -- (a.corner #3);
        \end{tikzpicture}
    }
    \begin{tikzpicture}[baseline={([yshift=-.5ex]current bounding box.center)}]
        \node[draw,thick,minimum size=3cm,regular polygon,regular polygon sides=6] (a) {};
        \coordinate (p1) at (60:2cm);
        \coordinate (p2) at (60*2:2cm);
        \coordinate (p3) at (60*3:2cm);
        \coordinate (p4) at (60*4:2cm);
        \coordinate (p5) at (60*5:2cm);
        \coordinate (p6) at (60*6:2cm);
        \draw (p1) node {\fpoo{1}{3}{4}{6}};
        \draw (p2) node {\fpoo{3}{1}{6}{4}};
        \draw (p3) node {\fpoo{3}{5}{6}{2}};
        \draw (p4) node {\fpoo{5}{3}{2}{6}};
        \draw (p5) node {\fpoo{5}{1}{2}{4}};
        \draw (p6) node {\fpoo{1}{5}{4}{2}};
        \draw ($(p1)!0.5!(p2)$) node {\fpo{1}{3}{4}{6}};
        \draw ($(p2)!0.5!(p3)$) node {\fp{3}{6}};
        \draw ($(p3)!0.5!(p4)$) node {\fpo{3}{5}{6}{2}};
        \draw ($(p4)!0.5!(p5)$) node {\fp{5}{2}};
        \draw ($(p5)!0.5!(p6)$) node {\fpo{5}{1}{2}{4}};
        \draw ($(p6)!0.5!(p1)$) node {\fp{1}{4}};
        \node[draw,minimum size=0.6cm,regular polygon,regular polygon sides=6] (a) {};
    \end{tikzpicture}
\end{center}

The definitive property of cyclohedron is its geometric factorization \cite{devadoss2002space} that every facet of a $n$-dimensional cyclohedron $W_n$ is the product of a $(n-i)$-dimensional cyclohedron $W_{n-i}$ and a $(i-1)$-dimensional associahedron $A_{i-1}$, i.e.
\[
    H\cong W_{n-i}\times A_{i-1},
\]
where $H$ is a facet of $W_n$ and $i$ is a positive integer.
In the following sections, we will see it again and again in different forms.

\subsection{Parke-Taylor Forms as Canonical Forms of Moduli Space Cyclohedra}

As shown in the last subsection, the interior of a cyclohedron $W_{n}(\sigma,\{\mathsf{s}_{i}\})$ is a 
simplex whose facets lay on the hyperplanes defined by equations 
\[
    \mathsf{s}_{\sigma(i)} z_{\sigma(i)}- \mathsf{s}_{\sigma(i+1)} z_{\sigma(i+1)}=0.
\]
Therefore, we can write down the Parke-Taylor form
as the canonical form \cite{Arkani-Hamed:2017tmz} of this simplex that
\begin{equation}
    \mathsf{PT}_n(
    \sigma(1)^{\mathsf s_{\sigma(1)}},\dots,
    \sigma(n)^{\mathsf s_{\sigma(n)}}):=\frac{\,2z_0\, \dd (\mathsf s_1z_1)\wedge \cdots \wedge\dd (\mathsf s_nz_n)}{(z_0- \mathsf s_{\sigma(1)} z_{\sigma(1)})(\mathsf s_{\sigma(1)} z_{\sigma(1)}- \mathsf s_{\sigma(2)} z_{\sigma(2)})\cdots(\mathsf s_{\sigma(n)} z_{\sigma(n)}+z_{0})}\:, \label{PTfactor}
\end{equation}
The appearance of $z_0$ on the numerator makes $\mathsf{PT}_n$ invariant under
$(\lambda,a):z\mapsto \lambda z^a$ for
all $(\lambda,a)\in (\cc-\{0\})\rtimes \zz_2$.
As revealed in \cite{Mizera:2017cqs}, Parke-Taylor forms can be rewritten as
\begin{equation*}
\label{dlogPT}
\begin{aligned}
\mathsf{PT}_n&(\sigma(1)^{\mathsf s_{\sigma(1)}},\dots,
\sigma(n)^{\mathsf s_{\sigma(n)}})\\
&=(-1)^{n}\operatorname{sgn}(\sigma)\,\dd \log \left( \frac{z_{0}-\mathsf s_{\sigma(1)} z_{\sigma(1)}}{\mathsf s_{\sigma(1)} z_{\sigma(1)}- \mathsf s_{\sigma(2)} z_{\sigma(2)}} \right)\wedge \cdots \wedge \dd \log \left( \frac{\mathsf s_{\sigma(n-1)} z_{\sigma(n-1)}- \mathsf s_{\sigma(n)} z_{\sigma(n)}}{\mathsf s_{\sigma(n)} z_{\sigma(n)}+z_{0}} \right),
\end{aligned}
\end{equation*}
so they are $\dd \log$ forms.

What's more, Parke-Taylor forms 
still work well after compactification (locally it's a series of blowing-ups), so
they are also the \textit{canonical forms} of cyclohedra%
\footnote{
More precisely, it is the pull-back of Parke-Taylor forms according to the compactification map that are the canonical forms of cyclohedra.
}%
. Since for every facet
of a cyclohedron, it can be factorized into a cyclohedron and an associahedron, the residue of 
the canonical form of cyclohedron at the facet should also be
factorized into the corresponding canonical forms.

To see it, consider the facet of 
$W_{n}(\operatorname{id},\{+\})$ whose label is 
$(i\cdots j)$. It corresponds to the blowing-up
of the pinch of the coordinates $\{z_i,\dots,z_j\}$ of the moduli space.
Let $z_k=z_i+ty_k$ for $i+1\leq k\leq j$, where $[y_{i+1},\dots,y_{j}]$ 
are projective coordinates of $\mathbb{P}^{j-i-1}$, then 
the residue of $\mathsf{PT}_n(1^+2^+\cdots n^+)$ at $t=0$ 
can be written as
\[
    \begin{aligned}
    \operatorname{Res}_{t=0}\mathsf{PT}_n(1^+2^+\cdots n^+)&=\operatorname{Res}_{t=0}\left[(-1)^{n}\dd \log \left( \frac{z_{0}- z_{1}}{ z_{1}-  z_{2}} \right)\wedge \cdots \wedge \dd \log \left( \frac{ z_{n-1}- z_{n}}{ z_{n}+z_{0}} \right)\right]\\
    &=(\pm) \,\mathsf{PT}_{n-(j-i)}(1^+2^+\cdots i^+(j+1)^+\cdots n^+)\wedge \Omega,
    \end{aligned}
\]
here $\mathsf{PT}_{n-(j-i)}$ is the canonical form of a \emph{cyclohedron}, while
\[
    \Omega=\dd \log \left( \frac{0- y_{i+1}}{y_{i+1}- y_{i+2}} \right)\wedge \cdots \wedge \dd \log \left( \frac{y_{j-2}-y_{j-1}}{y_{j-1}-y_j} \right)
\]
is the canonical form of an \emph{associahedron} up to a sign since it can be rewritten as
\[
    \Omega=\frac{(\pm)\,\dd y_{i+1}\wedge \cdots\wedge \dd y_{j-1}}{(0- y_{i+1})(y_{i+1}- y_{i+2})\cdots (y_{j-1}-1)}.
\]
in the piece of $\mathbb{P}^{j-i-1}$ where $y_j=1$.

Apparently, there are $2^{n}n!$ different Parke-Taylor forms 
corresponding to different cyclohedra. Like the case of associahedra, there are some linear relations between these Parke-Taylor forms. A simple relation is
\[
    \mathsf{PT}_n(
    \sigma(1)^{\mathsf s_{\sigma(1)}},\dots,
    \sigma(n)^{\mathsf s_{\sigma(n)}})
    =(-1)^n\mathsf{PT}_n(
    \sigma(n)^{-\mathsf s_{\sigma(n)}},\dots,
    \sigma(1)^{-\mathsf s_{\sigma(1)}}),
\]
some not so trivial relations are
\begin{equation}\label{ptr}
    \begin{aligned}
        \mathsf{PT}_n&(1^+2^{\mathsf s_2}\dots n^{\mathsf s_n})
        +\mathsf{PT}_n(2^{\mathsf s_2}1^+\dots n^{\mathsf s_n})
        +\cdots +
        \mathsf{PT}_n(2^{\mathsf s_2}\dots n^{\mathsf s_n}1^+)+
        \\
        &\mathsf{PT}_n(1^-2^{\mathsf s_2}\dots n^{\mathsf s_n})
        +\mathsf{PT}_n(2^{\mathsf s_2}1^-\dots n^{\mathsf s_n})
        +\cdots +
        \mathsf{PT}_n(2^{\mathsf s_2}\dots n^{\mathsf s_n}1^-)=0
    \end{aligned}
\end{equation}
and its permutations. The second kind of relations are analogues of 
\textit{Kleiss-Kuijf relations} in the case of associahedron. We conjecture that
the rank of these Parke-Taylor forms is $(2n-1)!!$ because of these relations.

To end this subsection, we express Parke-Taylor forms
in terms of $\{x_i\}$ and $\{\theta_i\}$. 
They live on the real line, so it's more convenient
for some purpose or other. In terms of $\{x_i\}$, 
Parke-Taylor forms are written as
\[
    \mathsf{PT}(
    \sigma(1)^{\mathsf s_{\sigma(1)}},\dots,
    \sigma(n)^{\mathsf s_{\sigma(n)}})=
    \frac{(x_0+x_0^{-1})\,\dd (\mathsf s_1x^{\mathsf s_1}_1)\wedge \cdots\wedge  \dd (\mathsf s_nx^{\mathsf s_n}_n)}{
    (x_0-\mathsf s_{\sigma(1)}x_{\sigma(1)}^{\mathsf s_{\sigma(1)}})\cdots 
    (\mathsf s_{{\sigma(n)}}x_{{\sigma(n)}}^{\mathsf s_{{\sigma(n)}}}+x_0^{-1})
    },
\]
where $s_{i}=\pm 1$.
Then it's direct to carry out Parke-Taylor forms
in terms of $\{\theta_i^{\mathsf s_i}\}$ by using 
$\mathsf s_ix_i^{\mathsf s_i}=\tan(\theta^{\mathsf s_i}/2)$,
\[
    \mathsf{PT}(
    \sigma(1)^{\mathsf s_{\sigma(1)}},\dots,
    \sigma(n)^{\mathsf s_{\sigma(n)}})=
    \frac{\sin(\theta^\mathsf s_{0,n+1}/2) \,\dd \theta_1\wedge \cdots \wedge \dd \theta_n}{\sin(\theta^\mathsf s_{0,\sigma(1)}/2)
    \cdots 
    \sin(\theta^\mathsf s_{\sigma(n-1),\sigma(n)}/2)\sin(\theta^\mathsf s_{\sigma(n),n+1}/2)},
\]
where $\theta^\mathsf s_{i,j}:=\theta_i^{\mathsf s_i}-\theta_{j}^{\mathsf s_{j}}$.

\subsection{Koba-Nielsen Factor and Scattering Equations: Particle Pairs on a Circle}

In order to connect to the kinematic space, we should introduce 
`Mandelstam variables' $s_{ij}$ for any pair of punctures with the usual properties, i.e. $s_{ii}=0$, $s_{ij}=s_{ji}$ 
and the conservation of momentum
\begin{equation}\label{mom_con}
    \sum_{j=0}^{2n+1}s_{ij}=0,
\end{equation}
but with additional identifications $s_{\tilde\imath \tilde\jmath}=s_{ij}$ since $z_i-z_j=0$ and $z_{\tilde \imath}-z_{\tilde \jmath}=0$ are the same hyperplane in the moduli space. It's natural to write down a factor
\begin{equation}\label{KN}
    \mathcal I=\prod_{0\leq i<j\leq 2n+1}(z_i-z_j)^{\alpha' s_{ij}}
\end{equation}
as the analogue of the so-called Koba-Nielsen factor in string theory. However, with the new identification $s_{ij}=s_{\tilde{\imath}\tilde{\jmath}}$ of Mandelstam variables, this factor $\mathcal I$ can be rewritten to drop this redundancy. For future use, it's convenient to rewrite it as
\begin{equation}\label{thetaKN}
    \mathcal I=\prod_{0\leq i<j\leq n}\sin\left(\frac{\theta_{i}-\theta_j}{2}\right)^{2\alpha' s_{ij}}
    \cos\left(\frac{\theta_{i}-\theta_j}{2}\right)^{2\alpha' s_{i\tilde \jmath}}
\end{equation}
in terms of $\{\theta_{i}\}$, or,
\begin{align}
    \mathcal{I} &=  \prod_{0\leq i<j\leq n} \Biggl(\frac{(x_{i}-x_{j})^{2}}{(1+x_{i}^{2})(1+x_{j}^{2})}\Biggr)^{s_{ij}} \Biggl(\frac{(1+x_{i}x_{j})^{2}}{(1+x_{i}^{2})(1+x_{j}^{2})}\Biggr)^{s_{i\tilde{\jmath}}}  \label{KNx}
\end{align}
in terms of $\{x_{i}\}$.

As shown in \cite{aomoto2011theory}, the number of stationary points of the Koba-Nielsen factor
equals to the number of associahedra in the moduli space $\mathcal M_{0,n}(\mathbb R)$, which is a highly non-trivial result because of the Poincar\'e dual. Therefore, 
it's believed that the number of stationary points of our new Koba-Nielsen
factor eq.\eqref{KN}
should equal the number of cyclohedra $2^{n-1}n!$. In the rest of this
subsection, we will prove that it's indeed the number of stationary points.

The stationary points of Koba-Nielsen factor is determined by 
the equation $\dd \log(\mathcal I)=0$,
then it gives a set of equations in $z$-variables, which are 
called \textit{scattering equations},
\begin{equation}
E_i:=\frac{\partial}{\partial z_i} \log(\mathcal I)
=\sum_{\tiny \substack{0\leq j\leq 2n+1\\ j\neq i}} \frac{s_{ij}}{z_i-z_{j}}=0\:, \label{scatteringequation} 
\end{equation}
for $i=1$, $\dots$, $n$. Or in terms of $\theta$-variables, 
from eq.\eqref{thetaKN}, the scattering equations are
\begin{equation}
    \sum_{\tiny \substack{0\leq j\leq n\\ j\neq i}}(s_{ij}\cot(\theta_{ij}/2)\theta_{i\tilde{\jmath}}-s_{i\tilde \jmath}\tan(\theta_{ij}/2))=0\:,
    \label{SEintheta}
\end{equation}
where $\theta_{ij}=\theta_i-\theta_j$.


This number $2^{n-1}n!$ can be recursively verified by taking soft limits~\cite{Cachazo:2013gna}. However, here we use another strategy proposed by Cachazo, Mizera and Zhang~\cite{Cachazo:2016ror}: each solution of scattering equation can be regarded as one stationary configuration of some particular particle system on some region of kinematic space $\mathcal{K}_{+}$.  
\begin{figure}[t]
\begin{center}
\includegraphics[scale=0.8]{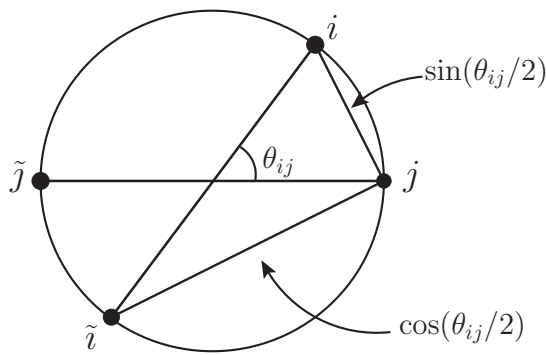}
\end{center}
\caption{Particle pairs on a circle} \label{CNSE}
\end{figure}

In our case, the region of kinematic space $\mathcal{K}_{+}$ is defined by all $s_{ij}$ and $s_{i\tilde{\jmath}}$ with $j\neq i$ are positive, this particle system is a $2$-dimensional system and consists of $n{+}1$ pairs of particles connected by $n{+}1$ light rigid rods of unit length (see figure \ref{CNSE}), where the centers of all rods are in the same position and each pair of particles connected by a rod have the same charge. If we denote the $i^{\textrm{th}}$ pair of particles as $i$ and $\tilde{\imath}$ and the charge product of particles $i$ and $j$ as $s_{ij}$, then the potential of this system is simply
\begin{equation}
V(\theta)= -2\sum_{0\leq i<j\leq n} \biggl( s_{ij}\log\bigl\lvert \sin(\theta_{ij}/2) \bigr\rvert+ s_{i\tilde{\jmath}}\log\bigl\lvert \cos(\theta_{ij}/2) \bigr\vert \biggr)
\:. \label{potential}
\end{equation}
Here we take no account of the all contribution from particle pairs connected by the rigid rods, and the factor 2 arise from the fact of $s_{ij}=s_{\tilde{\imath}\tilde{\jmath}}$ and $s_{i\tilde{\jmath}}=s_{\tilde{\imath}j}$. Since all particles mutually repel on the region $\mathcal{K}_{+}$, it is obvious that the solutions of scattering equation eq.\eqref{SEintheta} are stationary points of the potential eq.\eqref{potential} on this region, and all these solutions in terms of $\theta$-variables hence are real.

Next, we take a simple counting of stationary configurations of this system. A naive counting gives $2^{n+1}(n+1)!$ since these rods have two distinct endpoints and different orderings of these rods correspond different stationary configurations. However the symmetry $\operatorname{U}(1)\rtimes \zz_2$ mentioned above will reduce this number to $(2^{n}n!)/2$. (One rod can be taken as the reference point, and only one-half of all these orderings are indeed different since this is a central symmetric system.)

	\section{Intersection Numbers} \label{sec3}

In this section, we consider the new geometric structures of 
the moduli space $\mathcal M_{n+1}^{\mathrm{c}}(\rr)$ after introducing the 
Koba-Nielsen factor. 

Since Koba-Nielsen factor 
eq.\eqref{KN} or eq.\eqref{thetaKN} is not a well-defined 
single-valued function on $\mathcal M_{n+1}^{\mathrm{c}}(\rr)$ 
whose branch points are the hyperplains defined by $z_i-z_j=0$ for
all $0\leq i<j\leq 2n+1$, one should assign a branch of Koba-Nielsen factor on each chamber
for the sake of certainty. Therefore, 
we can define \textit{twisted} chambers
$\Delta\otimes \mathcal I_\Delta$ as our basic objects, where 
$\Delta$ is a chamber and $\mathcal I_\Delta$ is a given branch of $\mathcal I$ on $\Delta$. The same consideration still works for 
the compactification of $\mathcal M_{n+1}^{\mathrm{c}}(\rr)$. 

Generally, one can consider a twisted $m$-simplex
$\sigma\otimes \mathcal I_\sigma$,
where $\sigma$ is a topological $m$-simplex in $\mathcal M_{n+1}^{\mathrm{c}}(\rr)$,
then there is a natural inherited boundary map $\partial_{\mathcal I}$ for
any twisted simplex. 
For example, suppose $\sigma=\langle 01\cdots m\rangle$ is a topological $m$-simplex, where $\langle 01\cdots m\rangle$ is the standard notation \cite{Nakahara:2003nw}, the traditional boundary map is 
\[
	\partial^m \sigma=\sum_{k=0}^m (-1)^k\langle 01\cdots\hat{k}\cdots m\rangle,
\]
and the boundary map $\partial^m_{\mathcal I}$ for
any twisted simplex $\sigma\otimes \mathcal I_\sigma$ is
\[
	\partial_{\mathcal I}^m (\sigma\otimes \mathcal I_\sigma)=\sum_{k=0}^m (-1)^k\langle 01\cdots\hat{k}\cdots m\rangle \otimes I_{\sigma}|_{\langle 01\cdots\hat{k}\cdots m\rangle}.
\]
It's not hard to check that 
$\partial^{m-1}_{\mathcal I}\partial^{m}_{\mathcal I}=0$, so there's a 
natural homology group defined by 
\[
	H_m(\mathcal M_{n+1}^{\mathrm{c}}(\rr),\partial_{\mathcal I})
	=\frac{\ker \partial^{m-1}_{\mathcal I}}{\operatorname{im} \partial^{m}_{\mathcal I}}.
\]

Kita and Yoshida \cite{kita1994intersection,kita1994intersection2} investigated these objects for a general space and
a general factor. They developed a homology theory according to these
twisted simplexes and defined a new intersection number of two twisted 
cycles. Loosely speaking, the intersection number is a pairing
\[
	H_n(\mathcal M_{n+1}^{\mathrm{c}}(\rr),\partial_{\mathcal I})\times
	H_{n}(\mathcal M_{n+1}^{\mathrm{c}}(\rr),\partial_{\mathcal I})
	\to \mathbb{C}.
\]
Sabastian \cite{Mizera:2017cqs} showed that 
the intersection numbers of associahedra (with chosen
branches of Koba-Nielsen factor) on the moduli space form
the Kawai-Lewellen-Tye (KLT) kernel \cite{Kawai:1985xq} which relates
the tree-level open string amplitude and closed string amplitude. What's 
more, its $\alpha'\to 0$ limit gives the amplitude of 
bi-adjoint $\varphi^3$-theory  \cite{Arkani-Hamed:2017mur,Frost:2018djd} which relates to the 
Cachazo-He-Yuan (CHY) formula. In this section, 
we consider intersection numbers of cyclohedra. 

However, we will not go into the details of the whole geometric theory. Instead of defining intersection number step by step \cite{aomoto2011theory,kita1994intersection,kita1994intersection2}, we directly give a formula eq.\eqref{int_num} of two cyclohedra with the chosen branches $|\mathcal I|$.

Suppose $\Delta(\alpha)$ and $\Delta(\beta)$ are two 
cyclohedra, $K$ is their shared face with
the lowest codimension $k$, and they are on the different side
of $K$. Suppose $K = H_1 \cap \cdots \cap  H_k$ is the 
intersection of $k$ facets, and the chosen branches of $\mathcal I$
on $\Delta(\alpha)$ and $\Delta(\beta)$ is 
\[
	|\mathcal I| = f(z) \prod_{0\leq i<j\leq n}|z_i-z_j|^{2\alpha' s_{ij}}|z_i-z_{\tilde \jmath}|^{2\alpha' s_{i\tilde\jmath}},
\]
where the factor $f(z)$ will not bother us since it never vanishes on $\mathcal M_{n+1}^{c}(\rr)$. 
Then the intersection number \cite{aomoto2011theory,kita1994intersection,kita1994intersection2} of $\Delta(\alpha)$ and $\Delta(\beta)$ is defined by 
\begin{equation}\label{int_num}
	\langle\Delta(\alpha),\Delta(\beta) \rangle :=
	\left(\frac{1}{2\ii}\right)^k\prod_{i=1}^k\frac{1}{\sin(\alpha'\pi s_{H_i})}\sum_{\text{faces $H$ of $K$}}
	\frac{1}{d_H}
\end{equation}
where $s_{H_i}=2s_{i_1,\dots,i_k}=2\sum_{i_1\leq p<q\leq i_k}s_{pq}$ for $H_i$ whose 
label is $(i_1\cdots i_k)$ and 
\[
	d_{\bigcap_i H_{i}}=\prod_i \left(
		\exp(2\pi \ii \alpha' s_{H_i})-1\right) \quad
	\text{and} \quad d_{\Delta(\alpha)}=1.
\]

For example, 
\[
\begin{aligned}
	\Delta(1^+)&=
	\begin{tikzpicture}[scale=1,
		baseline={([yshift=-.5ex]current bounding box.center)}]
	\draw[thick] (0,0) -- (2,0);
	fill=black, inner sep=1pt,
	\filldraw[red,thick] (0,0) circle (1pt)
	node[left,black] {$(01)$};
	\filldraw[red,thick] (2,0) circle (1pt)
	node[right,black] {$(0\tilde 1)$};
\end{tikzpicture}\\
	\Delta(1^+2^+)&=
	\begin{tikzpicture}[scale=1,
		baseline={([yshift=-.5ex]current bounding box.center)}]
		\draw[thick] (0,0.2) -- (0,1.8) -- (0.2,2) --
		(1.8,2) -- (2,1.8) -- (0.2,0) -- (0,0.2);
		\draw (1,1) node[below=8pt,right=8pt] {$(12)$};
		\draw (1,2) node[above] {$(0\tilde 2)$};
		\draw (0,1) node[left] {$(01)$};
		\draw[thick,red] (0,1.8) -- (0.2,2);
		\draw[thick,red] (1.8,2) -- (2,1.8);
		\draw[thick,red] (0.2,0) -- (0,0.2);
		\draw[thick,red,->] (-0.3,2.1) node[left,black] {$(01\tilde 2)$}
		.. controls (-0.2,2.2) and (-0.1,2.1) .. (0.05,1.95);
		\draw[thick,red,->] (-0.3,-0.1) node[left,black] {$(012)$}
		.. controls (-0.2,-0.2) and (-0.1,-0.1) .. (0.05,0.05);
		\draw[thick,red,->] (2.3,2.1) node[right,black] {$(\tilde 012)$}
		.. controls (2.1,2.2) and (2.05,2.1) .. (1.95,1.95);
	\end{tikzpicture}\quad
	\Delta(1^+2^-)=
	\begin{tikzpicture}[scale=1,
		baseline={([yshift=-.5ex]current bounding box.center)}]
		\draw[thick] (0,0.2) -- (0,1.8) -- (0.2,2) --
		(1.8,2) -- (2,1.8) -- (0.2,0) -- (0,0.2);
		\draw (1,1) node[below=8pt,right=8pt] {$(1\tilde 2)$};
		\draw (1,2) node[above] {$(02)$};
		\draw (0,1) node[left] {$(01)$};
		\draw[thick,red] (0,1.8) -- (0.2,2);
		\draw[thick,red] (1.8,2) -- (2,1.8);
		\draw[thick,red] (0.2,0) -- (0,0.2);
		\draw[thick,red,->] (-0.3,2.1) node[left,black] {$(012)$}
		.. controls (-0.2,2.2) and (-0.1,2.1) .. (0.05,1.95);
		\draw[thick,red,->] (-0.3,-0.1) node[left,black] {$(01\tilde 2)$}
		.. controls (-0.2,-0.2) and (-0.1,-0.1) .. (0.05,0.05);
		\draw[thick,red,->] (2.3,2.1) node[right,black] {$(0\tilde 12)$}
		.. controls (2.1,2.2) and (2.05,2.1) .. (1.95,1.95);
	\end{tikzpicture}
\end{aligned}
\]
and the self intersection number of $\Delta(1^+)$ is
\[
\begin{aligned}
	\langle \Delta(1^+),\Delta(1^+)\rangle &=
	1+\frac{1}{\exp(2\pi \ii \alpha' 2s_{01})-1}+
	\frac{1}{\exp(2\pi \ii \alpha' 2s_{0\tilde 1})-1}\\
	&=\frac{1}{2\ii}\left(\frac{1}{\tan(\alpha'\pi 2s_{01})}
	+\frac{1}{\tan(\alpha'\pi 2s_{0\tilde 1})}\right).
\end{aligned}
\]
For $\Delta(1^+2^+)$ and $\Delta(1^+2^-)$, they intersect at 
the face labeled by $2s_{01}$, so
\[
\begin{aligned}
	\langle \Delta(1^+2^+),\Delta(1^+2^-)\rangle
	&=\frac{1}{2\ii}\frac{1}{\sin(\alpha'\pi 2s_{01})}
	\left(
	1+\frac{1}{\exp(2\pi \ii \alpha' 2s_{012})-1}+
\frac{1}{\exp(2\pi \ii \alpha' 2s_{01\tilde 2})-1}
	\right)\\
	&=\frac{1}{(2\ii)^2}\frac{1}{\sin(\alpha'\pi 2s_{01})}
	\left(
\frac{1}{\tan(\alpha'\pi 2s_{012})}
	+\frac{1}{\tan(\alpha'\pi 2s_{01\tilde 2})}
	\right).
\end{aligned}
\]

Finally, define the `KLT inverse matrix' by 
\[
	m_{\alpha'}(\alpha\vert \beta)=(2\ii)^{n}
	\langle \Delta(\alpha),\Delta(\beta)\rangle,
\]
where $n=\dim \Delta(\alpha)=\dim \Delta(\beta)$, and define
\begin{equation}
	m(\alpha\vert \beta)=\lim_{\alpha'\to 0}(2\alpha'\pi)^n
	m_{\alpha'}(\alpha\vert\beta)\:.  \label{eqvof2m}
\end{equation}
For example,
\[
m(1^+2^+,1^+2^-)
=\frac{1}{s_{01}}\left(
	\frac{1}{s_{012}}+\frac{1}{s_{01\tilde 2}}
\right).
\]
In the next section, we will use CHY formula to calculate $m(\alpha\vert\beta)$, which gives another representation.

At the end of this section, we briefly comment on the dual of $H_n(\mathcal M_{n+1}^{c}(\rr),\partial_{\mathcal I})$, which is the cohomology group on $\mathcal M_{n+1}^{c}(\rr)$ defined by a flat connection
\[
	\nabla=\dd + \omega_{\mathcal I}\wedge,
\]
with the property $\nabla^2=0$, where 
\[
	\omega_{\mathcal I}=\dd \log\mathcal{I}=
	\sum_{i=1}^n E_i \dd z_i=
	\sum_{i=1}^n\sum_{\tiny \substack{0\leq j\leq 2n+1\\ j\neq i}} \frac{s_{ij}}{z_i-z_{j}}\dd z_i
\]
is a well-defined single-valued form on $\mathcal M_{n+1}^{\mathrm{c}}(\rr)$. 
Let's denote these cohomology groups by 
$H^{\bullet}(\mathcal M_{n+1}^{\mathrm{c}}(\rr),\omega_{\mathcal I})$.
As shown in \cite{esnault1992cohomology}, $H^{n}(\mathcal M_{n+1}^{\mathrm{c}}(\rr),\omega_{\mathcal I})$ is generated
by $\dd \log$ forms
\[
	\dd \log \left( \frac{z_{i_1}\pm z_{j_1}}{z_{k_1}\pm z_{l_1}} \right)\wedge \cdots \wedge \dd \log \left( \frac{z_{i_n}\pm z_{j_n}}{z_{k_n}\pm z_{l_n}} \right),
\]
where $0\leq i_r,j_r,k_r,l_r\leq n$ for 
all $1\leq r\leq n$, which should be invariant under $z\mapsto 1/z$.
Parke-Taylor forms eq.\eqref{dlogPT} are in these forms.
Unlike the case of $\mathcal M_{0,n}(\mathbb R)$, the independent 
Parke-Taylor forms can \emph{no longer} generate the whole space
$H^n(\mathcal M_{n+1}^{\mathrm{c}}(\rr),\omega_\mathcal I)$. One has to consider more forms, 
for example,
\[
	\dd \log \left(\frac{z_1}{z_0}\right)\wedge 
	\dd \log \left(\frac{z_2}{z_0}\right)
	=\frac{\dd z_1\wedge \dd z_2}{z_1z_2}.
\]
It is invariant under the transformation $z\mapsto 1/z$,
and then belongs to $H^2(\mathcal M_{3}^{c}(\rr),\omega_{\mathcal I})$. However, it is not the canonical form of a cyclohedron anymore.
    \section{Kinematic Cyclohedra and CHY Formula} \label{sec4}
In \cite{Arkani-Hamed:2017mur}, Arkani-Hamed, Bai, He and Yan have constructed associahedra in kinematic space. They found a set of equations and inequalities
of $X$-variables
\[
    X_{ij}:=s_{i,i+1,\dots,j-1}=\sum_{i\leq k<l\leq j-1}s_{kl},
\]
where the interior of an associahedron is given by $\{X_{ij}>0\}$ and every facet represented by $X_{ij}$ is laying on the hyperplane defined by $X_{ij}=0$. They also found that the scattering equations can be interpreted as a map between worldsheet associahedra and kinematic associahedra.
Here, we generalize these stories to cyclohedra. 

\subsection{Kinematic Cyclohedra}

Suppose $P_n$ is a $(2n+2)$-gon, for every diagonal with vertices $i$ and $j$, we put a variable $X_{ij}$ on it, 
where labels are living in $\zz_{2n+2}$. Thanks to momentum conservation and our identification of 
$s_{ij}$ and $s_{\tilde\imath\tilde\jmath}$, $X_{ij}=X_{ji}=X_{\tilde \imath\tilde \jmath}$.
Usually, it's more convenient to use $X_{\tilde \imath\tilde \jmath}
=X_{ij}$ to represent a pair of centrally symmetric
diagonals $(i,j)$ and $(\tilde\imath,\tilde\jmath)$, where we 
think the longest diagonals $(i,\tilde \imath)=(\tilde \imath,i)$ as `degenerate' pairs. Therefore, there are $n(n+1)$
independent $X$-variables. 
The inverse map from $X$-variables to
$s$-variables is given by the following 
identity
\[
    -s_{ij}=X_{ij}+X_{i+1,j+1}-X_{i+1,j}-X_{i,j+1}.
\]
For future use, we generalize the above identity to
\begin{equation}\label{sXmap}
    -\sum_{i\leq a<j<k\leq b<l}s_{ab}=X_{ik}+X_{jl}-X_{jk}-X_{il}.
\end{equation}

In order to construct a $n$-dimensional object from $n^2+n$
independent variables, we need $n^2$ constraints, which can be chosen like the old story~\cite{Arkani-Hamed:2017mur}: let
\begin{equation}\label{constraints}
    c_{ij}=-s_{ij}=X_{i,j}+X_{i+1,j+1}-X_{i,j+1}-X_{i+1,j}
\end{equation}
be positive constants for $0\leq i<n$ and $i+1<j\leq 2n$, 
where we still require that $c_{i\tilde\jmath}$ and $c_{j\tilde\imath}$ are the same, so the rank of these constraints is
$n^2$.
It's very convenient to take into account equations of $c_{ij}$
for $n< i\leq 2n$ by identifying $c_{\tilde \imath\tilde \jmath}$ 
and $c_{ij}$. All in all, we require that $s_{ij}=-c_{ij}<0$ for 
\[
    0\leq i<j-1<j\leq 2n, \quad i\neq n \quad \text{and}\quad 
    j\neq \tilde{n}.
\]
Our polytope $\mathcal{W}_n$ is given by inequalities
\[
    X_{ij}\geq 0 \quad \text{for all $0\leq i<j\leq 2n+1$}\:,
\]
and setting eq.\eqref{constraints} to be positive constants. Such construction can also obtained from $C_{n}$ cluster algebra by using the method proposed in~\cite{bazier2018abhy}.
We will show that it is a cyclohedron after a short example.

For $n=2$, independent variables are
\(
    \{
        X_{02},X_{0\tilde 0},X_{0\tilde 1},X_{1\tilde 1},
        X_{1\tilde 2}, X_{2\tilde 2}
    \},
\)
and constraints are
\begin{align*}
    X_{02}+X_{0\tilde 1}-X_{0\tilde 0}&=c_{02},\\
    X_{0\tilde 0}+X_{1\tilde 1}-2X_{0\tilde 1}&=c_{0\tilde 0},\\
    X_{0\tilde 1}+X_{1\tilde 2}-X_{1\tilde 1}&=c_{0\tilde 1},\\
    X_{1\tilde 1}+X_{2\tilde 2}-2X_{1\tilde 2}&=c_{1\tilde 1},
\end{align*}
where the factor $2$ arise form identification and the equation for $c_{1\tilde 0}$ is the same as the equation for $c_{0\tilde 1}$. 
Thus, $\mathcal{W}_2$ is the polytope defined by inequalities
\[
    X_{02}\geq 0,\quad X_{0\tilde 0}\geq 0,
    \quad X_{0\tilde 1}\geq 0,\quad X_{1\tilde 1}\geq 0,
    \quad X_{1\tilde 2}\geq 0,\quad  X_{2\tilde 2}\geq 0.
\]
It's shown in figure \ref{w2} by using $X_{0\tilde{0}}$ and $X_{02}$ as the coordinates. Therefore, $\mathcal{W}_2$ is 
a hexagon, that is a 2 dimensional cyclohedron. Now let's consider the $\mathcal{W}_n$ for general $n$. 

\begin{figure}[htbp]
\begin{center}
    \begin{tikzpicture}[scale=1.5]
    \filldraw[fill=blue!15] (0,0) -- (0,1) -- (1,2) -- (3,3)
    -- (4,3) -- (4,0) -- (0,0);
    \draw[thick] (0,0) -- (0,1) -- (1,2) -- (3,3)
    -- (4,3) -- (4,0) -- (0,0);
    \draw[thick,->] (0,0) -- (5,0) node[right] {$X_{0\tilde 0}$};
    \draw[thick,->] (0,0) -- (0,3.5) node[above] {$X_{02}$};
    \draw (2,0) node[below] {$X_{02}$};
    \draw (2,2.6) node[above] {$X_{1\tilde 1}$};
    \draw (3.5,3) node[above] {$X_{1\tilde 2}$};
    \draw (4,1.5) node[right] {$X_{2\tilde 2}$};
    \draw (0.4,1.5) node[above] {$X_{0\tilde 1}$};
    \draw (0,0.5) node[left] {$X_{0\tilde 0}$};
\end{tikzpicture}
\end{center}
\caption{The $2$-dimensional kinematic cyclohedron $\mathcal{W}_2$ after
taking $c_{02}=c_{0\tilde 0}=c_{0\tilde 1}=c_{1\tilde 1}=1$.}
\label{w2}
\end{figure}
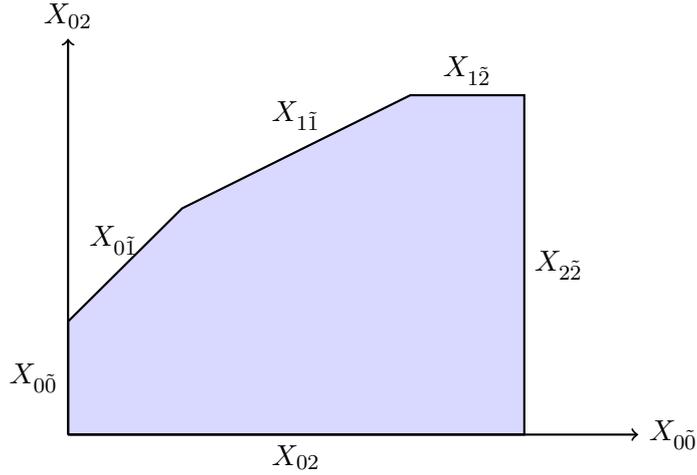

\begin{pro}
    $\mathcal{W}_n$ is the $n$-dimensional cyclohedron.
\end{pro}

\begin{proof}
In our construction, every facet of the polytope $\mathcal{W}_n$ is given by setting some
particular $X_{ij}\to 0$. However, it's impossible to set arbitrary two $X$-variables approaching to $0$ together. Indeed, recall eq.\eqref{sXmap}, for $j\leq n$ and $l\leq \tilde n$,
\[
    X_{jk}+X_{il}=X_{ik}+X_{jl}-\sum_{i\leq a<j<k\leq b<l}c_{ab},
\]
if $X_{ik}$, $X_{jl}\to 0$, the right-hand side of above equation will
be negative, but the left-hand side will never be negative, so there is a contradiction. 
It follows that we cannot reach a lower dimensional boundary of $\mathcal{W}_n$ by setting two intersecting diagonals to zero. 
Conversely, lower dimensional boundaries of $\mathcal{W}_n$ can only be reached by setting non-intersecting diagonals to zero. 
On the other hand,
the dissections of $(2n+2)$-gon by setting non-intersecting diagonals to zero are always centrally symmetric since $X_{ij}=X_{\tilde{\imath}\tilde{\jmath}}$. Therefore, $\mathcal{W}_n$ is a $n$-dimensional cyclohedron according to Definition \ref{def-1}.
\end{proof}

It's very natural to construct other cyclohedra $\mathcal{W}(\sigma,\{\mathsf{s}_{i}\})$ in kinematic space. The data $\{\mathsf{s}_1,\mathsf{s}_2,\ldots,\mathsf{s}_n\}$ and permutation $\sigma$ define an centrally symmetric ordering in $S^{2n+2}$, and vice versa. For example,
$(2^+1^-3^+) \leftrightarrow (02\tilde 13\tilde 0\tilde 2 1\tilde 3)$.
Therefore, we identify these two notations of centrally symmetric orderings. Now, we can defined $X$-variables in ordering $\alpha$ by setting
\[
	X_{\alpha(i),\alpha(j)}=s_{\alpha(i),\alpha(i+1),\dots,\alpha(j-1)}.
\]
Then $\mathcal W_n(\alpha)$ in kinematic space is similarly given by setting 
\begin{equation}
    H_n(\alpha)=\{c_{\alpha(i)\alpha(j)} =\text{positive constant}\:\vert\: 0\leq i<j-1<j\leq 2n,\:i\neq n \text{ and }j\neq \tilde{n} \}
\end{equation}
and
\begin{equation}
    X_{\alpha(i),\alpha(j)}\geq 0 \quad \text{for all}\: i<j. 
\end{equation}

Cyclohedra are simple polytopes, i.e. every vertex of $n$-dimensional cyclohedron $W_n$ is adjacent to $n$ facets. It's direct from the 
Definition \ref{def-1} of cyclohedron. Since every centrally symmetric triangulation of a $(2n+2)$-gon, which
corresponds to a vertex of $W_n$, has $n$ pairs of centrally symmetric diagonals, which correspond to the
adjacent facets of this vertex. 
Therefore the canonical form \cite{Arkani-Hamed:2017tmz} of 
$\mathcal{W}_n$ is 
\begin{equation}\label{scattering_form}
    \Omega(\mathcal{W}_n)=\sum_{\text{vertex $Z$}} \operatorname{sign}(Z)\bigwedge_{a=1}^n
    \dd\log X_{i^Z_a,j^Z_a},
\end{equation}
where vertex $Z$ corresponds to a centrally symmetric triangulation of this $(2n+2)$-gon, $X_{i^Z_a,j^Z_a}$ is the $X$-variable which corresponds to the pair of diagonals $(i^Z_a,j^Z_a)$ and $(\tilde\imath^Z_a,\tilde\jmath^Z_a)$ in this triangulation, and $\operatorname{sign}(Z)$ is evaluated on
the ordering of the facets in the wedge product
to ensure $\Omega(\mathcal{W}_n)$ is a projective form \cite{Arkani-Hamed:2017tmz,Arkani-Hamed:2017mur}, for example,
\begin{align*}
    \Omega(\mathcal{W}_1) &= \frac{\dd X_{0\tilde{0}}}{X_{0\tilde{0}}}- \frac{\dd X_{1\tilde{1}}}{X_{1\tilde{1}}} \:,  \\
    \Omega(\mathcal{W}_2) &= \frac{\dd X_{02}\wedge\dd X_{0\tilde{0}}}{X_{02}X_{0\tilde{0}}} -\frac{\dd X_{02}\wedge\dd X_{2\tilde{2}}}{X_{02}X_{2\tilde{2}}} +\frac{\dd X_{1\tilde{2}}\wedge\dd X_{2\tilde{2}}}{X_{1\tilde{2}}X_{2\tilde{2}}} \\
    &\quad- \frac{\dd X_{1\tilde{2}}\wedge\dd X_{1\tilde{1}}}{X_{1\tilde{2}}X_{1\tilde{1}}}+ \frac{\dd X_{0\tilde{1}}\wedge\dd X_{1\tilde{1}}}{X_{0\tilde{1}}X_{1\tilde{1}}} -\frac{\dd X_{0\tilde{1}}\wedge\dd X_{0\tilde{0}}}{X_{0\tilde{1}}X_{0\tilde{0}}} \:. 
\end{align*}

On the other aspect, every vertex of $\mathcal{W}_{n}$ also corresponds to the dual graph of each centrally symmetric triangulation of a $(2n{+}2)$-gon, see figure \ref{w3}. 
\begin{figure}[htbp]
    \begin{center}
        \begin{tikzpicture}[baseline={([yshift=-.5ex]current bounding box.center)}]
            \node[draw,thick,minimum size=3cm,regular polygon,regular polygon sides=6,gray] (a) {};
            \coordinate (p0) at (a.corner 1);
            \coordinate (pt2) at (a.corner 2);
            \coordinate (pt1) at (a.corner 3);
            \coordinate (pt0) at (a.corner 4);
            \coordinate (p2) at (a.corner 5);
            \coordinate (p1) at (a.corner 6);
            \coordinate (t1) at ($(p0)!0.5!(pt1)!0.25!(pt2)$);
            \coordinate (t2) at ($(pt0)!0.5!(p1)!0.25!(p2)$);
            \coordinate (t3) at ($(pt0)!0.5!(pt1)!0.5!(0,0)$);
            \coordinate (t4) at ($(p0)!0.5!(p1)!0.5!(0,0)$);
        
            \draw[gray] (p0) node[above] {$0$};
            \draw[gray] (pt2) node[above] {$\tilde 2$};
            \draw[gray] (pt1) node[left] {$\tilde 1$};
            \draw[gray] (pt0) node[left] {$\tilde 0$};
            \draw[gray] (p2) node[right] {$2$};
            \draw[gray] (p1) node[right] {$1$};
            \draw[thick,gray] (p0) -- (pt1);
            \draw[thick,gray] (p1) -- (pt0);
            \draw[thick,gray] (p0) -- (pt0);
            \draw[thick,darkred] (30:-2) node[left] {$\tilde{0}$}
            -- (t3) -- (t4);
            \draw[thick,darkred] (30:2) node[right] {$0$}
            -- (t4);
            \draw[thick,darkred] (5*30:2) node[left] {$\tilde{1}$}
            -- (t1) -- (t3);
            \draw[thick,darkred] (5*30:-2) node[right] {$1$}
            -- (t2) -- (t4);
            \draw[thick,darkred] (3*30:2) node[above] {$\tilde 2$} -- (t1);
            \draw[thick,darkred] (3*30:-2) node[below] {$2$} -- (t2);
        \end{tikzpicture}
        \qquad \qquad
        \begin{tikzpicture}[baseline={([yshift=-.5ex]current bounding box.center)}]
            \node[draw,thick,minimum size=3cm,regular polygon,regular polygon sides=6] (a) {};
            \filldraw[gray] (0,0) circle [radius=1pt];
            \draw (a.corner 1) node[above] {$0$};
            \draw (a.corner 2) node[above] {$\tilde 1$};
            \draw (a.corner 3) node[left] {$\tilde 2$};
            \draw (a.corner 4) node[left] {$\tilde 0$};
            \draw (a.corner 5) node[right] {$1$};
            \draw (a.corner 6) node[right] {$2$};
            \draw[thick] (a.corner 2) -- (a.corner 6);
            \draw[thick] (a.corner 3) -- (a.corner 6);
            \draw[thick] (a.corner 3) -- (a.corner 5);
        \end{tikzpicture}
    \end{center}
\caption{The dual graph of a triangulation (left) and a triangulation in the ordering $(021\tilde 0\tilde 2\tilde 1)$ (right).}
\label{w3}
\end{figure}
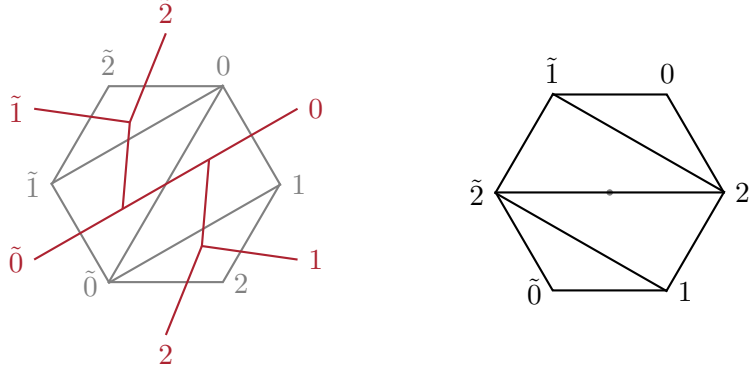
If we further think dual graphs as planar Feynman diagrams with scalar
propagators, they share the same poles with the term corresponding to the vertex $Z$ 
in eq.\eqref{scattering_form}
\[
    \bigwedge_{a=1}^n
    \dd\log X_{i^Z_a,j^Z_a}=\frac{\bigwedge_{a=1}^n\dd X_{i^Z_a,j^Z_a}}{\prod_{a=1}^n X_{i^Z_a,j^Z_a}}.
\]
For example, for the triangulation in the left-hand side of figure \ref{w3}, poles of this planar Feynman diagram are $s_{12}=s_{\tilde 1\tilde 2}$
and $s_{012}=s_{\tilde 0\tilde 1\tilde 2}$, and the corresponding scattering form is
\[
    \frac{\dd X_{02}\wedge \dd X_{03}}{X_{02}X_{03}}=\frac{\dd s_{01}\wedge\dd s_{012}}{s_{01}s_{012}}.
\]
Note that, $1/(s_{01}s_{012})$ is not the amplitude determined by this Feynman diagram.
Besides, we can even read $\operatorname{sign}(Z)$ from this viewpoint. Indeed, if $Z$ and $Z'$ are full centrally symmetric triangulations related by a mutation of a pair of diagonals, then
$\operatorname{sign}(Z)=-\operatorname{sign}(Z')$.

We can extend this viewpoint of dual graph to other orderings. For an ordering $\alpha$, define
$\mathsf{csp}(\alpha)$ to be the set of all dual graphs of full centrally symmetric triangulations in
the ordering $\alpha$ (see the right-hand side of figure.\ref{w3} for an example). 
If a planar graph $g$ belongs to $\mathsf{csp}(\alpha)$, we say that this graph is compatible with
$\alpha$.
Therefore, 
\[
    \Omega(\mathcal{W}_n(\alpha))=\sum_{g\in \mathsf{csp}(\alpha)} \operatorname{sign}(g)\bigwedge_{a=1}^n
    \dd\log X_{i^g_a,j^g_a}.
\]
Now for any graph $g$, we can pullback it to subspace $H(\alpha)$,
then \cite{Arkani-Hamed:2017mur}
\begin{equation}\label{17}
    \left.\left(\operatorname{sign}(g)\bigwedge_{a=1}^n
    \dd X_{i^g_a,j^g_a}\right)\right|_{H(\alpha)}=
    \begin{cases}
        \dd^n X(\alpha) & \text{if $g\in \mathsf{csp}(\alpha)$},\\
        0&\text{otherwise},
    \end{cases}
\end{equation}
where we assume that $\operatorname{sign}(g)$ for a compatible graph in different
ordering are equal. One can check that $\dd^n X(\alpha)$ is independent of choice of $g$, so it's the volume form of $\mathcal W_n(\alpha)$.

Now we define $m(\alpha\vert\beta)$ by the pullback of $\Omega(\mathcal{W}_n(\alpha))$
at the space $H_n(\beta)$:
\[
    \Omega(\mathcal{W}_n(\alpha))|_{H_n(\beta)} =: m(\alpha\vert \beta) \,\dd^n X(\beta),
\]
or more concretely by eq.\eqref{17},
\begin{equation} \label{malphabeta}
    m(\alpha\vert \beta)=\sum_{g\in \mathsf{csp}(\alpha)\cap \mathsf{csp}(\beta)}
    \prod_{a=1}^n\frac{1}{X_{i^g_a,j^g_a}} \:.
\end{equation}
    
To this point, one can easily verify eq.(\ref{eqvof2m}) by checking the combinatoric structure of both sides.

\subsection{Scattering Equations as a Map Between Cyclohedra}
Up to now, we have constructed cyclohedra in 
moduli space and kinematic space. Here we claim
that the connection between these two cyclohedra
is given by the so-called \textit{scattering equation map} \cite{Arkani-Hamed:2017mur} which is 
the solution of $s_{ij}$ of the scattering equations eq.\eqref{scatteringequation}.

Recall that, in our construction of kinematic cyclohedron, 
\[
    s_{01},s_{12},\dots,s_{n-1,n},s_{0\tilde n},s_{1\tilde n},\dots,s_{n\tilde n}
\]
are not constants. We can eliminate $s_{i\tilde n}$ by the conservation of momentum, 
then scattering equations eq.\eqref{scatteringequation} are linear equations
of $\{s_{i,i+1}\,\vert \,0\leq i\leq n-1\}$, whose solution is 
\begin{equation}\label{sm}
    s_{k,k+1}=\sum_{\substack{0\leq i\leq k<k+1\leq j\leq 2n\\ j\neq i+1}}\frac{\sin \left(\theta_{i\tilde n}/2\right) \sin \left(\theta_{j\tilde n}/2\right) \sin \left(\theta_{k,k+1}/2\right)}{\sin \left(\theta_{ij}/2\right) \sin \left(\theta_{k\tilde n}/2\right) \sin \left(\theta_{k+1,\tilde n}/2\right)}c_{ij}
\end{equation}
for $k=0,\dots,n-1$. Eq.(\ref{sm}) defines a map $\varphi:z\mapsto X$ from moduli space to kinematic space. 
It's direct from eq.(\ref{sm}) that $X_{k,k+2}=s_{k,k+1}>0$ when 
$0<\theta_1<\cdots<\theta_n<\pi$. 
The other $X_{ij}$'s can also be computed by the
definition of $X$-variables and one can check 
that all these $X$-variables are positive.
Therefore, this map further maps the worldsheet cyclohedra $W_{n}(\operatorname{id},\{+\})$ into the kinematic cyclohedra $\mathcal{W}_{n}$.

What's more, the scattering equation map is \emph{a map between the interiors of these two cyclohedra}. %
It's equivalent to say that $\varphi(z)$ is in the boundary of $\mathcal{W}_n$ if and only if 
$z$ is in the boundary of $W_n(\operatorname{id},\{+\})$. %
This is straightforward from the scattering equation map eq.\eqref{sm} or scattering equations eq.\eqref{scatteringequation}. 
For example, suppose $\theta_p=\theta_r+tx_p$ for $r<p\leq s$,
when $t$ goes to $0$, 
\[
    s_{k,k+1}\to \sum_{i=r}^{k-1} 
    \frac{x_{k+1}-x_k}{x_{k+1}-x_i}c_{i,k+1}+\sum_{i=r}^{k} \sum _{j=k+2}^{s} \frac{x_{k+1}-x_k}{x_j-x_i}c_{ij},
\]
for $r\leq k<s$. Therefore,
\[
    \sum_{k=r}^{s-1}s_{k,k+1}\to \sum_{i=r}^s\sum_{j=i+2}^sc_{ij},
\]
and then
\[
   X_{r,s+1}=s_{r,\dots,s}=\sum_{k=r}^{s-1}s_{k,k+1}-\sum_{i=r}^s\sum_{j=i+2}^sc_{ij}\to 0.
\]
Conversely, one can use eq.\eqref{sXmap} to rewrite eq.\eqref{scatteringequation} in $X$ and 
$z$-variables, then when $X_{ij}$ goes to zero, the corresponding $z_k-z_i$ goes to 
zero for $i<k<j$. Actually, it does not depend on the condition $c_{ij}>0$ at all. We conjecture that 
$\varphi$ is a one-to-one map after imposing these conditions.

The pushforward \cite{Arkani-Hamed:2017mur}
of scattering equation map 
connects the canonical forms of these two cyclohedra, i.e. 
\[
    \sum_{\text{sol. }z}\mathsf{PT}_n(1^+\cdots n^+)=\Omega(\mathcal{W}_n(\operatorname{id}))|_{H_n(\operatorname{id})}
    =m(\operatorname{id}\vert\operatorname{id})\,\dd^n X(\operatorname{id}),
\]
where $z$ are the solution of scattering equations. One can consider 
other ordering pairs $\alpha$, $\beta$ and the scattering equation map $\varphi^\alpha$, then
\[
    \sum_{\text{sol. }z}\mathsf{PT}_n(\beta)=\Omega(\mathcal{W}_n(\operatorname{\alpha}))|_{H_n(\beta)}
    =m(\alpha\vert \beta)\,\dd^n X(\beta).
\]
We can rewrite pushforward in delta function form: 
\[
\begin{aligned}
    m(\alpha\vert\beta)&=\int \mathsf{PT}_n(\beta)\prod_{a=1}^n \delta(X_{\alpha(i_a),\alpha(j_a)}-\varphi^\alpha_a(z))\\
    &=\int \mathsf{PT}_n(\beta)\mathsf{PT}_n(\alpha)\prod_{a=1}^n \delta(E_a)\\
    &=: m_{\text{CHY}}(\alpha\vert \beta).
\end{aligned}
\]
    \section{$Z$-integrals on the Moduli Space \texorpdfstring{$\mathcal{M}^{\mathrm{c}}_{n+1}(\mathbb{R})$}{M(R)}} \label{sec5}

It is the place to consider some basic prototypes of string-like integrals on such moduli space since the Parke-Taylor factors and Koba-Nielsen factors on this space have been defined.  
In analogy with the story in $\mathcal{M}_{0,n}$, it's natural to define the so-called `$Z$-integrals' \cite{Mafra:2016mcc} on $\mathcal{M}_{n+1}^{\mathrm{c}}(\mathbb{R})$ as 
\begin{equation}
    Z_{\alpha}(\beta)=\int_{\mathfrak{C}(\alpha)} \mathsf{PT}(\beta) \prod_{0\leq i<j\leq n}
    \left(-\frac{(z_{\alpha(i)}-z_{\alpha(j)})^2}{4z_{\alpha(i)}z_{\alpha(j)}}\right)^{\alpha^{\prime}s_{ij}}\left(\frac{(z_{\alpha(i)}+z_{\alpha(j)})^2}{4z_{\alpha(i)}z_{\alpha(j)}}\right)^{\alpha^{\prime}s_{i\tilde \jmath}}   \label{Zintegral}
\end{equation}
where $\mathsf{PT}(\beta)$ is the Parke-Taylor factor defined by eq.(\ref{PTfactor}) with respect to the ordering $\beta$, $\mathfrak{C}(\alpha)$ denotes the integration ordering on a circle which is defined by $\theta_{\alpha(0)}< \theta_{\alpha(1)}<\cdots<\theta_{\alpha(n)}<\theta_{\alpha(\tilde{0})}$ (recall that $\theta_{i}=\arg(z_{i})$), and the form of Koba-Nielsen factor is directly derived form eq.(\ref{thetaKN}) by a change of variables and hence ensure its reality. Obviously, these integral satisfy KK-like relations
\begin{equation}\label{ZintKK}
    \begin{split}
        Z_{\alpha}(1^+2^{\mathsf s_2}\dots n^{\mathsf s_n})
        +&Z_{\alpha}(2^{\mathsf s_2}1^+\dots n^{\mathsf s_n})
        +\cdots +
        Z_{\alpha}(2^{\mathsf s_2}\dots n^{\mathsf s_n}1^+)+
        \\
        &Z_{\alpha}(1^-2^{\mathsf s_2}\dots n^{\mathsf s_n})
        +Z_{\alpha}(2^{\mathsf s_2}1^-\dots n^{\mathsf s_n})
        +\cdots +
        Z_{\alpha}(2^{\mathsf s_2}\dots n^{\mathsf s_n}1^-)=0
    \end{split}
\end{equation} In what follows, we will limit ourselves to the case of $\mathfrak{C}(012\cdots n)$ since all the other cases just differ by a relabelling of indices, and the subscript $\alpha$ on $Z$ will also be omitted since the integration orderings have been fixed. Further more, we fix $z_{0}$ as $0$ (then $\theta_{0}=0$ and $\theta_{\tilde{0}}=\pi$) to mod out the $U(1)$ redundancy.

In this section, we will take $\alpha^{\prime}$-expansion for two simple examples of integrals \eqref{Zintegral} to give a taste of number theoretical properties. We will see, not only multiple zeta values, but also \emph{alternating Euler sums} appear in the $\alpha^{\prime}$-expansion. Before proceeding, it's worth mentioning an important limit of Koba-Nielson factor \eqref{KNx}  
\begin{align}
       \mathcal{I}^{\prime} &=  \prod_{0\leq i<j\leq n} \Biggl(\frac{(x_{i}-x_{j})^{2}}{(1+x_{i}^{2})(1+x_{j}^{2})}\Biggr)^{s_{ij}}
       = \prod_{0\leq i<j\leq n} (x_{i}-x_{j})^{2\alpha^{\prime }s_{ij}} \prod_{i=0}^{n}\lvert x_{i}-\ii\rvert^{2\alpha^{\prime}s_{i\tilde{\imath}}} \label{KNoc}
\end{align}
which is obtained by taking all $s_{i\tilde{\jmath}}$ with $i\neq j$  to be vanish, where $s_{i\tilde{\imath}}$ now is ${-}\sum_{j=0}^{n}s_{i j}$. Note that, eq.(\ref{KNoc}) is proportional to the Koba-Nielsen factor for $n{+}1$ open strings and one closed string~\cite{Stieberger:2009hq}
\begin{equation}
    \mathcal{I}_{n+1,1}=\lvert z-\bar{z}\rvert^{\alpha^{\prime}q_{\parallel}^{2}}\prod_{i<j}\lvert x_{i}-x_{j}\rvert^{2\alpha^{\prime}p_{i}\cdot p_{j}}\prod_{i}\lvert x_{i}-z\rvert^{2\alpha^{\prime}p_{i}\cdot q}
\end{equation}
under the gauge fixing $z=\ii$. However, some $Z$-integrals we defined don't have a good behaviour under this limit, and a comprehensive treatment of such integrals lie somewhat out the main line of this paper, we leave it to the future work.

Let us first consider the simplest cases, $n=1$. In this case, there is actually only one kind of Parke-Taylor factor since $\mathsf{PT}(1^{+})=-\mathsf{PT}(1^{-})$, and the $Z$-integral is simply
\begin{align}
     Z(1^{+})&=\int_{0}^{\infty} \frac{\dd x_{1}}{-2x_{1}} \biggl(\frac{x_{1}^{2}}{1+x_{1}^{2}}\biggr)^{\alpha^{\prime}s_{01}}
        \biggl(\frac{1}{1+x_{1}^{2}}\biggr)^{\alpha^{\prime}s_{0\tilde{1}}} \nonumber \\
         &=-\frac{\Gamma(\alpha^{\prime}s_{01})\Gamma(\alpha^{\prime}s_{0\tilde{1}})}{4\Gamma(\alpha^{\prime}s_{01}+\alpha^{\prime}s_{0\tilde{1}})} \:. \label{intn=1} 
\end{align}
Here the $x$-parameterization is used. The result is the usual Beta function, hence the $\alpha^{\prime}$-expansion can be trivially obtained. The reader can check this result is divergent under the limit $s_{0\tilde{1}}\to 0$.

The first nontrivial case is $n=2$. In this case, the result of integral can no longer be expressed as the special functions we are familiar with, like the usual hypergeometric function appearing in string tree amplitudes. Thus, we turn to seek its $\alpha^{\prime}$-expansion.  As pointed out in the paper \cite{Broedel:2013tta}, the obstacle to this expansion is singularities in Mandelstam variables pole. Putting it manifestly, the leading order of the integral (\ref{Zintegral}), which can be obtained as a CHY formula
\begin{equation}
    Z_{\alpha}(\beta)\sim (\alpha^{\prime})^{-n}m(\alpha\vert\beta)+\cdots \:, \label{leadofint}
\end{equation}     
can not appear in a Taylor expansion in $\alpha^{\prime}$. An approach to this problem is pole subtractions introduced in \cite{Schlotterer:2018zce}. The basic idea of this method is to subtract the integrand by several terms which can be integrated out easier and contain the pole contributions. There are four $Z$-integrals for $n=2$ which are related by 
\begin{equation}
    Z(1^{+}2^{+})=-Z(2^{+}1^{+})-Z(1^{-}2^{+})-Z(2^{+}1^{-})\:,
\end{equation}   
and the pole subtractions of individual terms on R.H.S. are easier than $Z(1^{+}2^{+})$, which can be seen from leading contributions of these integrals. Further more, $Z(1^{-}2^{+})$ and $Z(2^{+}1^{-})$ can be obtained from $Z(2^{+}1^{+})$ through shifting particle index by $-1$ and $+1$, respectively. Thus, the remaining task is to attack the integral $Z(2^{+}1^{+})$.

We omit the detail of pole subtractions since this process is technical and tedious. The $\alpha^{\prime}$-expansion of remainder term after pole subtractions can be obtained by using the Maple program HyperInt~\cite{Panzer:2014caa} with an appropriate choice of the integration ordering. The result is simply
\begin{align}
    -Z(2^{+}1^{+})&=\frac{1}{4 \alpha^{\prime 2}}\biggl(\frac{1}{ s_{12} s_{012}}+\frac{1}{s_{12}s_{\tilde{0}12}}\biggr) 
   +\frac{\zeta_{2}}{4}\left(4-\frac{4 s_{\tilde{0}2}}{s_{\tilde{0}12}}-\frac{s_{\tilde{0}12}+s_{012}}{s_{12}}-\frac{4 s_{01}}{s_{012}}\right) \nonumber \\
    &\quad + \frac{\alpha^{\prime}\zeta_{3}}{2}\Biggl(\frac{(s_{\tilde{0}12}+s_{012})^2}{2 s_{12}}+\frac{4 s_{01} (s_{12}+s_{01})}{s_{012}}+\frac{4 s_{\tilde{0}2} (s_{12}+s_{\tilde{0}2})}{s_{\tilde{0}12}} \nonumber \\
    &\qquad\qquad -2 (2 s_{12}+s_{\tilde{0}2}+s_{01}+s_{01\tilde{2}})\Biggr) +O(\alpha^{\prime 2}) \:. \label{a'expansion}
\end{align} 
The next and higher orders of $\alpha^{\prime}$ are too long to record here, thus we just report the numbers appearing in the $\alpha^{\prime}$-expansion: 
\begin{align*}
    &\alpha^{\prime2}:\qquad \zeta_{2}^{2}\:, \\
    &\alpha^{\prime3}:\qquad \zeta_{2}\zeta_{3},\:\zeta_{5} \\
    &\alpha^{\prime4}:\qquad \zeta_{1,-3}\zeta_{2},\zeta_{2}^{3},\zeta_{3}^{2}
\end{align*}
where the convention for MZVs is
\[
    \zeta_{n_{1},\ldots,n_{r}}= \operatorname{Li}_{\lvert n_{1}\rvert,\ldots,\lvert n_{r}\rvert}\biggl(\frac{n_{1}}{\lvert n_{1}\rvert},\ldots ,\frac{n_{r}}{\lvert n_{r}\rvert}\biggr)
\]
with the usual Goncharov's polylogarithms $\operatorname{Li}_{n_{1},\ldots,n_{r}}(z_{1},\ldots,z_{r})$ of weight $n=n_{1}+\cdots+ n_{r}$~\cite{Goncharov:1998kja}.

Here we briefly remark the $\alpha^{\prime}$-expansion of $Z(2^{+}1^{+})$ on the following points:
\begin{enumerate}[(i)]
    \item We do this integral beginning with $x$-parameterization. Then the result is not as simple as in eq.(\ref{a'expansion}). Where the numerical coefficients involving the letters $\{\pm\ii\}$ appear. (`Letters' mean the possible numbers appearing in iterated integral which can be translated into the input of polylogarithms.)
    \item However, a huge simplification of the result is realized by a series of changes of integration variables, $\{y_{1}=x_{1}^{2},y_{2}=x_{1}x_{2}\}$ and $\{t_{i}=y_{i}/(1+y_{i})\vert i=1,2\}$.
    In terms of the variables $\{t_{i}\}$, the Koba-Nielsen factor can be written as 
    \[
    \mathcal{I}=t_{1}^{\alpha^{\prime}s_{01\tilde{2}}}(1-t_{1})^{\alpha^{\prime}s_{0\tilde{1}2}}\, t_{2}^{2\alpha^{\prime}s_{02}}(1-t_{2})^{\alpha^{\prime}s_{0\tilde{2}}}(t_{1}-2t_{1}t_{2}+t_{2}^{2})^{\alpha^{\prime}s_{2\tilde{2}}}(t_{2}-t_{1})^{2\alpha^{\prime}s_{12}}\:.   
    \]
    In its present form, the linear reducibility of $\alpha^{\prime}$-expansion of the Koba-Nielsen factor is manifest, then a similar analysis as in \cite{Brown:2009qja} will tell us that only the letters $\{-1,0,1\}$ can appear, namely, the numbers appearing in the $\alpha^{\prime}$-expansion can only be the usual multiple zeta values and alternating Euler sums, like $\log 2$ and $\zeta_{1,-3}$ and so on. Note that, it is the introduction of $y$-variables such that the linear reducibility better.
   \item There is no contribution of order $\alpha^{\prime-1}$ although the number $\log 2=\zeta_{-1}$ of weight 1 does exist. 
\end{enumerate}  

    \section{Outlook}

In this article, we study a new moduli space $\mathcal{M}_{n+1}^{\mathrm{c}}$ and the corresponding kinematic space via two mathematic tools, intersection theory and positive geometry. We have shown that $\mathcal{M}_{n+1}^{\mathrm{c}}(\mathbb{R})$ is tiled by $2^{n-1}n!$ copies of the  cyclohedron $W_{n}$, and each of which corresponds to a stationary point of the system with $n{+}1$ pairs of particles on a circle. We construct the `inverse KLT matrix' for this case by considering intersection numbers of twisted cycles. At the same time, we construct another cyclohedron $\mathcal{W}_{n}$ in kinematic space
and prove the scattering equations provide a map from $W_{n}$ to $\mathcal{W}_{n}$. It seems that everything works well as in the old story $\mathcal{M}_{0,n}$. However, there are several problems we have not taken up yet, for instance, how to triangulate the kinematic cyclohedra themselves as in \cite{He:2018svj}, from which we expect to obtain a recursion relation for $m(\alpha\vert\beta)$. 

Meanwhile, there are two aspects are obviously different from the case for $\mathcal{M}_{0,n}$. One is the mismatching between the number of independent Parke-Taylor factors, which is conjectured to be $(2n-1)!!$, and the number of solutions of scattering equations, which is proved to be $(2n)!!/2$. The reason for this mismatching, as pointed out at the end of section \ref{sec3}, %
is that there are $\dd\log$ forms beyond Parke-Taylor forms in the whole twisted cohomology group. One consequence is the lack of BCJ-like relations as in~\cite{Cachazo:2012uq,Cachazo:2013gna} which first appear in gauge theory amplitudes and is crucial for the color-kinematics duality~\cite{Bern:2008qj}. Therefore, an important and interesting question is how to find other $\dd\log$ forms to complete the basis of the whole twisted cohomology group.

Another direction is obviously the $Z$-integrals we defined in section \ref{sec5}, which suppose to be a special case of a very rich class of `integrals over cluster associahedron'~\cite{Arkani}, which will be explored (together with many other topics) further. We just showed the $\alpha^{\prime}$-expansion for two simple examples, however, new numerical coefficients comparing with the $Z$-integrals over $\mathcal{M}_{0,n}$ have started to appear. So, two natural questions are which kind of numbers can appear in the $\alpha^{\prime}$-expansion of such integrals, do these $\alpha^{\prime}$-expansions have the similar structure as in~\cite{Schlotterer:2012ny,Stieberger:2014hba,Brown:2018omk}? Meanwhile, it is also interesting to find the $\alpha^{\prime}$-expansion without integrating as in~\cite{Broedel:2013aza,Mafra:2016mcc}.

Apart from all questions mentioned above, another interesting question is the connection with the scattering of open strings and one closed string. Although the corresponding Koba-Nielsen factor emerges through a specific limit of Mandelstam variables, other ingredients of such string integrals are still missing in our treatment. 
What's more, it's important to find the connection between the moduli spaces of such string integrals, i.e. one puncture in the interior and $n{+}1$ punctures at the boundary of a disk, and our case. The connection between these two moduli spaces may be a clue of all missing parts.

\section*{Acknowledgement}The original idea for this work came from the study of positive geometries, canonical forms and integrals related to cluster associahedra, by Nima Arkani-Hamed, Song He and Hugh Thomas, to whom we are grateful for suggesting the project, sharing ideas and many valuable inputs. We thank Gongwang Yan for collaborations in an early stage of the project. We also thank Giulio Salvatori for inspiring discussions, and Erik Panzer for his kind instruction on the nice program HyperInt.

	% ...
	\bibliographystyle{utphys}
	\bibliography{main.bib}
\end{document}